\documentclass{amsart}
\usepackage{amsmath,amsfonts,amssymb,amsthm,floatrow,caption,subcaption, algorithm, algorithmic, graphics,enumerate}
\usepackage[colorlinks,citecolor=blue,urlcolor=blue]{hyperref}
\usepackage{tikz}
\usetikzlibrary{positioning}
 \usepackage[round]{natbib}
\input xy
\xyoption{all}
\newtheorem{thm}{Theorem}[section]
\newtheorem{lem}[thm]{Lemma}
\newtheorem{prop}[thm]{Proposition}
\newtheorem{cor}[thm]{Corollary}

\theoremstyle{definition}
\newtheorem{defn}[thm]{Definition}

\newtheorem{exmp}[thm]{Example}

\theoremstyle{remark}
\newtheorem{rem}[thm]{Remark}

\renewcommand{\emptyset}{\varnothing}
\newcommand{\bi}{\leftrightarrow}
\newcommand{\Left}{\text{Left}}
\newcommand{\Right}{\text{Right}}
\newcommand{\QED}{\ifhmode\unskip\nobreak\fi\quad {\rm Q.E.D.}} % QED

\newcommand{\cF}{\mathcal{F}}

\newcommand{\cO}{\mathcal{O}}

\newcommand{\R}{\mathbb{R}}
\newcommand{\cT}{\mathcal{T}}

\numberwithin{equation}{section}

\title[Identifiability of Structural Equation Models]{Generic Identifiability of Linear Structural Equation Models by Ancestor Decomposition}

\author{Mathias Drton} 
\address{Department of Statistics, University
  of Washington, Seattle, WA, U.S.A.}
\email{md5@uw.edu}

\author{Luca Weihs} 
\address{Department of Statistics, University
  of Washington, Seattle, WA, U.S.A.}
\email{lucaw@uw.edu}

\date{\today}                                

\begin{document}

\begin{abstract}
Linear structural equation models, which relate random variables via linear interdependencies and Gaussian noise, are a popular tool for modeling multivariate joint distributions.  These models correspond to mixed graphs that include both directed and bidirected edges representing the linear relationships and correlations between noise terms, respectively. A question of interest for these models is that of parameter identifiability, whether or not it is possible to recover edge coefficients from the joint covariance matrix of the random variables.  For the problem of determining generic parameter identifiability, we present an algorithm that extends an algorithm from prior work by Foygel, Draisma, and Drton (2012).  The main idea underlying our new algorithm is the use of ancestral subsets of vertices in the graph in application of a decomposition idea of Tian (2005).
\end{abstract}


\maketitle

\section{Introduction}

It is often useful to model the joint distribution of a random vector $X=(X_1,...,X_n)^T$ in terms of a collection of noisy linear interdependencies. In particular, we may postulate that each $X_w$ is a linear function of $X_1,...,X_{w-1},X_{w+1},...,X_n$ and a stochastic noise term $\epsilon_w$. Models of this type are called \emph{linear structural equation models} and can be compactly expressed in matrix form as
\begin{align}
	X = \lambda_0 + \Lambda^TX + \epsilon \label{eq:sem_eq}
\end{align}
where $\Lambda=(\lambda_{vw})$ is a $n\times n$ matrix, $\lambda_0 = (\lambda_{01},...,\lambda_{0n})^T\in\R^n$, and $\epsilon = (\epsilon_1,...,\epsilon_n)^T$ is a random vector of error terms. We will adopt the classical assumption that $\epsilon$ has a non-degenerate multivariate normal distribution with mean 0 and covariance matrix $\Omega=(\omega_{vw})$. With this assumption it follows immediately that $X$ has a multivariate normal distribution with mean $(I-\Lambda)^{-T}\lambda_0$ and covariance matrix
\begin{align} \label{eq:sigma}
	\Sigma = (I-\Lambda)^{-T}\Omega(I-\Lambda)^{-1}
\end{align}
where $I$ is the $n\times n$ identity matrix.  We refer the reader to the book by \citet{bollen} for background on these types of models.

We obtain a collection of interesting models by imposing different patterns of zeros among the coefficients in $\Lambda$ and $\Omega$. These models can then be naturally associated with mixed graphs containing both directed and bidirected edges. In particular, the graph will contain the directed edge $v\to w$ when $\lambda_{vw}$ is not required to be zero and, similarly, will include the bidirected edge $v\bi w$ when $\omega_{vw}$ is potentially non-zero. Representations of this type are often called \emph{path diagrams} and were first advocated for in \cite{wright1921, wright1934}.

A natural question arising in the study of linear structural equations is that of identifiability; whether or not it is possible to uniquely recover the two parameter matrices $\Lambda$ and $\Omega$ from the covariance matrix $\Sigma$ they define via~\eqref{eq:sigma}.  The most stringent version, known as global identifiability, amounts to unique recovery of every pair $(\Lambda,\Omega)$ from the covariance matrix $\Sigma$.  This global property can be characterized efficiently  \citep{drton2011}.  Often, however, a less stringent notion that we term generic identifiability is of interest.  This property requires only that a generic (or randomly chosen) pair $(\Lambda,\Omega)$ can be recovered from its covariance matrix.  The computational complexity of deciding whether a given mixed graph $G$ defines a generically identifiable linear structural equation model is unknown.  There are, however, a number of graphical criteria that are sufficient for generic identifiability and can be checked in polynomial time in the number of considered variables (or vertices of the graph).  To our knowledge, the most widely applicable such criterion is the Half-Trek Criterion (HTC) of \citet*{halftrek}, which built on earlier work of \citet{brito06}.  The HTC also comes with a necessary condition for generic identifiability but in this paper our focus is on the sufficient condition.  We remark that an extension of the HTC for identification of subsets of  edge coefficients is given in \cite{chen2014}.

We begin with a brief review of background such as the formal connection between structural equation models and mixed graphs, and give a review of prior work in Section \ref{sec:pre}.  In the main Section \ref{sec:anc_decomp}, we will demonstrate a simple method by which to infer generic identifiability of certain entries of $(\Lambda,\Omega)$ by examining subgraphs of a given mixed graph $G$ that are induced by ancestral subsets of vertices. This will extend the applicability of the HTC in the case of acyclic mixed graphs after applying the decomposition techniques of \citet{tian}.  We leverage this extension in an efficient algorithmic form. In Section \ref{sec:comp}, we report on computational experiments demonstrating the applicability of our findings.   A brief conclusion is given in Section \ref{sec:conclusion}.

\section{Preliminaries} \label{sec:pre}

We assume that the reader is familiar with graphical representations of structural equation models and thus only provide a quick review of these topics. For a more in-depth treatment see, for instance, \cite{pearl2009} or \cite{halftrek}. 

\subsection{Mixed Graphs}

For any $n\geq 1$, let $[n] := \{1,...,n\}$. We define a \emph{mixed graph} to be a triple $G=(V,D,B)$ where $V=[n]$ is a finite set of vertices and $D,B\subset V\times V$.  The sets $D$ and $B$ correspond to the directed and the bidirected edges, respectively. When $(v,w)\in D$, we will write $v\to w\in G$ and if $(v,w)\in B$ then we will write $v\bi w\in G$. Since edges in $B$ are bidirected the set $B$ is symmetric, that is, we have $(v,w)\in B \iff (w,v)\in B$. We require that both the directed part $(V,D)$ and bidirected part $(V,B)$ contain no self loops so that $(v,v)\not\in D\cup B$ for all $v\in V$. If the directed graph $(V,D)$ does not contain any cycles, so that there are no vertices $v,w_1,...,w_m\in V$ such that $v\to w_1,w_1\to w_2,..., w_m\to v \in G$, then we say that $G$ is \emph{acyclic}; note, in particular, that $G$ being acyclic does not imply $(V,B)$ does not contain any (undirected) cycles.

A \emph{path} from $v$ to $w$ is any sequence of edges from $D$ or $B$ beginning at $v$ and ending with $w$, here edges need not obey direction and loops are allowed. A \emph{directed path} from $v$ to $w$ is then any path from $v$ to $w$ all of whose edges are directed and pointed in the same direction, away from $v$ and towards $w$. Finally, a \emph{trek} $\pi$ from a \emph{source} $v$ to a \emph{target} $w$ is any path that has no colliding arrowheads, that is, $\pi$ must be of the form
\begin{align*}
	v^{L}_{l}\leftarrow v^{L}_{l-1} \leftarrow ...\leftarrow v^{L}_{0} \longleftrightarrow v^{R}_{0} \to v^{R}_{1} \to ... \to v^{R}_{r-1} \to v^{R}_{r}
\end{align*}
or
\begin{align*}
	v^{L}_{l}\leftarrow v^{L}_{l-1} \leftarrow ...\leftarrow v^{L}_{1} \leftarrow v^T \to v^{R}_{1} \to ... \to v^{R}_{r} \to v^{R}_{r}
\end{align*}
where $v^L_l = v$, $v^R_r = w$, and we call $v^T$ the \emph{top} node. If $\pi$ is as in the first case then we let $\Left(\pi) = \{v^{L}_{0},...,v^{L}_{l}\}$ and $\Right(\pi) = \{v^{R}_{0},...,v^{R}_{r}\}$, if $\pi$ is as in the second case then we let $\Left(\pi) = \{v^T, v^{L}_{1},...,v^{L}_{l}\}$ and $\Right(\pi) = \{v^T, v^{R}_{1},...,v^{R}_{r}\}$. Note that, in the second case, $v^T$ is included in both $\Left(\pi)$ and $\Right(\pi)$. A trek $\pi$ is called a \emph{half-trek} if $|\Left(\pi)| = 1$ so that $\pi$ is of the form
\begin{align*}
	v^{L}_{0} \longleftrightarrow v^{R}_{0} \to v^{R}_{1} \to ... \to v^{R}_{r-1} \to v^{R}_{r}
\end{align*}
or
\begin{align*}
	v^T \to v^{R}_{1} \to ... \to v^{R}_{r} \to v^{R}_{r}
\end{align*}

It will be useful to reference the local neighborhood structure of the graph. For this purpose, for all $v\in V$, we define the two sets
\begin{align}
	 & pa(v) = \{w\in V:w\to v\in G\} , \\
%	\text{the children of $v$} &= ch(v) = \{w:v\to w\in G\} \\
	& sib(v) = \{w\in V:w\bi v\in G\}.
\end{align}
The former comprises the \emph{parents} of $v$, and the latter contains the \emph{siblings} of $v$.

%\subsubsection{Representing Structural Equation Models with Mixed Graphs}

We associate a mixed graph $G$ to a linear structural equation model as follows. Let $\R^D$ be the set of real $n\times n$ matrices $\Lambda=(\lambda_{vw})$ with support $D$, i.e., $\lambda_{vw}\not=0 \implies (v,w)\in D$.     Let $PD_n$ be the cone of $n\times n$ positive definite matrices $\Omega=(\omega_{vw})$.  Define $PD(B)\subset PD_n$ to be the subset of positive definite matrices with support $B$, i.e. for $v\not=w$, $\omega_{vw}\not=0 \implies v\bi w\in G$.

In this paper, we focus on acyclic graphs $G$.  If $G$ is acyclic then the matrix $I-\Lambda$ is invertible for all $\Lambda\in\R^D$.  In other words, the equation system from~\eqref{eq:sem_eq} can always be solved uniquely for $X$.  We are led to the following definition.

\begin{defn}
The \emph{linear structural equation model given by an acyclic mixed graph} $G=(V,D,B)$ with $V=[n]$ is the collection of all $n$-dimensional normal distributions with covariance matrix
\begin{align*}
	\Sigma = (I-\Lambda)^{-T}\Omega(I-\Lambda)^{-1}
\end{align*}
for a choice of $\Lambda \in \R^D$ and $\Omega\in PD(B)$.
\end{defn}

\subsection{Prior Work and the HTC} \label{sec:htc}

For a fixed acyclic mixed graph $G$,  let $\Theta := \R^D\times PD(B)$ be the parameter space and $\phi_G:\Theta\to PD_n$ be the map
\begin{align}
	\phi_G:(\Lambda,\Omega) \mapsto (I-\Lambda)^{-T}\Omega(I-\Lambda)^{-1}.
\end{align}
Then the question of identifiability is equivalent to asking whether the \emph{fiber}
\begin{align*}
	\cF(\Lambda,\Omega) := \phi_G^{-1}(\{\phi_G(\Lambda,\Omega)\})%\{(\Lambda',\Omega')\in\Theta: \phi(\Lambda',\Omega') = (\Lambda,\Omega)\}
\end{align*}
equals the singleton $\{(\Lambda,\Omega)\}$.   We note that the above notions are well-defined also when $G$ is not acyclic but, in that case, $\R^D$ should be restricted to contain only matrices $\Lambda$ with $I-\Lambda$ invertible.

When $\cF(\Lambda,\Omega)=\{(\Lambda,\Omega)\}$ for all $(\Lambda,\Omega)\in\Theta$, so that $\phi_G$ is injective on $\Theta$, then $G$ is said to be \emph{globally identifiable}.  Global identifiability is, however, often too strong a condition.  So-called instrumental variable problems, for instance, give rise to graphs $G$ that would not be globally identifiable but for which the set of $(\Lambda, \Omega)$ on which identifiability fails has measure zero;  see the example in the introduction of \citet{halftrek}.   Instead, we will be concerned with the question of generic identifiability.

\begin{defn}
A mixed graph $G$ is said to be \emph{generically identifiable} if there exists a proper algebraic subset $A\subset \Theta$ so that $G$ is identifiable on $\Theta\setminus A$.
\end{defn}

Here, as usual, an \emph{algebraic set} is defined as the zero-set of a collection of polynomials.   We again refer the reader to the introduction of \cite{halftrek} for  an in-depth exposition on why generic identifiability is an often appropriate weakening of global identifiability.  

Now there will be cases when we are interested in understanding the generic identifiability of certain coefficients of a mixed graph $G$ rather than all coefficients simultaneously. In these cases we say that the coefficient $\lambda_{vu}$ (or $\omega_{vu}$),  for $u,v\in V$, is generically identifiable in $G$ if the projection of the fiber $\cF(\Lambda,\Omega)$ onto $\lambda_{vu}$ (or $\omega_{vu}$) is a singleton for all $\Theta\setminus A$ where $A\subset \Theta$ is a proper algebraic set. 

Let $\Lambda$ and $\Omega$ be matrices of indeterminates as in Equation \eqref{eq:sigma} with zero pattern corresponding to $G$. Then, by the Trek Rule of  \citet{wright1921}, see also \citet*{spirtes2000}, the covariance $\Sigma_{vw}$ can be represented as a sum of monomials corresponding to treks between $v$ and $w$ in $G$. 
%To see why this is the case note that
%\begin{align} \label{eq:inverse_directed}
%	((I-\Lambda)^{-1})_{vw} = \sum_{\pi\in\cP} \prod_{x\to y\in \pi} \lambda_{xy}
%\end{align}
%where $\cP(v,w)$ is the collection of all directed paths from $v$ to $w$. Equation \ref{eq:inverse_directed} follows by writing $(I-\Lambda)^{-1} = I + \Lambda+\Lambda^2+...$. When $G$ is acyclic we have that $\Lambda$ is strictly upper triangular so that $ I + \Lambda+\Lambda^2+... =  I + \Lambda+\Lambda^2+...+\Lambda^{n-1}$, in the general case some conditions must be placed on $\Lambda$ so that convergence occurs.
To state the Trek Rule formally, let $\cT(v,w)$ be the set of all treks from $v$ to $w$ in $G$. Then for any $\pi\in \cT(v,w)$, if $\pi$ contains no bidirected edge and has top node $z$, we define the \emph{trek monomial} as
\begin{align*}
	\pi(\Lambda,\Omega) = \omega_{zz}\prod_{x\to y\in \pi} \lambda_{xy},
\end{align*}
and if $\pi$ contains a bidirected edge connecting $u,z\in V$ then we define the trek monomial as
\begin{align*}
	\pi(\Lambda,\Omega) = \omega_{uz}\prod_{x\to y\in \pi} \lambda_{xy}
\end{align*}
We may then state the rule as follows.
\begin{prop}[Trek Rule] \label{prop:trek_rule}
For all $v,w\in V$, the covariance matrix $\Sigma = (I-\Lambda)^{-T}\Omega(I-\Lambda)^{-1}$ corresponding to a mixed graph $G$ satisfies
\begin{align*}
	\Sigma_{vw} = \sum_{\pi\in \cT(v,w)} \pi(\Lambda,\Omega).
\end{align*}
\end{prop}

Before giving the statement of the HTC, we must first define what is meant by a half-trek system. Let $\Pi = \{\pi_1,...,\pi_m\}$ be a collection of $m$ treks with each $\pi_i$ having source $x_i$ and target $y_i$, then $\Pi$ is called a \emph{system of treks} from $X=\{x_1,...,x_m\}$ to $Y=\{y_1,...,y_m\}$ if $|X|=|Y|=m$ so that all sources as well as all targets are pairwise distinct.  If each $\pi_i$ is a half-trek, then $\Pi$ is a \emph{system of half-treks}. Moreover, a collection $\Pi=\{\pi_1,...,\pi_m\}$ of treks is said to have \emph{no sided intersection} if
\begin{align*}
	\Left(\pi_i)\cap\Left(\pi_j) =\emptyset = \Right(\pi_i)\cap\Right(\pi_j), \ \forall i\not=j
\end{align*}

Let $htr(v)$ be the collection of vertices $w\in V\setminus (\{v\}\cup sib(v))$ for which there is a half-trek from $v$ to $w$, these $w$ are called \emph{half-trek reachable} from $v$.  We have the following definition and result of \citet{halftrek}.
\begin{defn}
A set of nodes $Y\subset V$ satisfies the \emph{half-trek criterion} with respect to a node $v\in V$ if
\begin{enumerate}[(i)]
	\item $|Y|=|pa(v)|$,
	\item $Y\cap (\{v\}\cup sib(v)) = \emptyset$, and
	\item there is a system of half-treks with no sided intersection from $Y$ to $pa(v)$.
\end{enumerate}
\end{defn}
 %With this definition we may now present their main identifiability result.
\begin{thm}[HTC-identifiability] \label{thm:htc-id}
Let $(Y_v:v\in V)$ be a family of subsets of the vertex set $V$ of a mixed graph $G$. If, for each node $v$, the set $Y_v$ satisfies the half-trek criterion with respect to $v$, and there is a total ordering $\prec$ on the vertex set $V$ such that $w\prec v$ whenever $w\in Y_v\cap htr(v)$, then $G$ is rationally identifiable.
\end{thm}

The assertion that $G$ is \emph{rationally identifiable} means that the inverse map $\phi_G^{-1}$ can be represented as a rational function on $\Theta\setminus A$ where $A$ is some proper algebraic subset of $\Theta$. Clearly, rational identifiability is a stronger condition than generic identifiability. If a graph $G$ satisfies Theorem \ref{thm:htc-id} we will say that $G$ is \emph{HTC-identifiable} (HTCI). In a similar vein, Theorem 2 of \citet{halftrek} gives sufficient conditions for a graph $G$ to be generically unidentifiable (with generically infinite fibers of $\phi_G$), and we will call such graphs \emph{HTC-unidentifiable} (HTCU). Graphs that are neither HTCI or HTCU are called \emph{HTC-inconclusive}, these are the graphs on which progress is left to be made.

As is noted in Section 8 in \cite{halftrek} we may extend the power of the HTC by using the graph decomposition techniques of \citet{tian}. Let $C_1,...,C_k\subset V$ be the unique partitioning of $V$ where $v,w\in C_i$ if and only if there exists a (possibly empty) path from $v$ to $w$ composed of only bidirected edges.  In other words, $C_1,...,C_k$ are the connected components of $(V,B)$, the bidirected part of $G$. For $i=1,\dots,k$, let 
\begin{align*}
	V_i &= C_i\cup pa(C_i), & D_i &= \{v\to w\in G: v\in V_i,\ w\in C_i\}, \\
	B_i & = \{v\bi w\in G: v,w\in C_i\}, & G_i &= (V_i,D_i,B_i).
\end{align*}
Then the mixed graphs $G_1,...,G_k$ are called the \emph{mixed components} of $G$. From the work of \citet{tian}, \cite{halftrek} present the following theorem.
\begin{thm}[Tian Decomposition] \label{thm:tian}
For an acyclic mixed graph $G$ with mixed components $G_1,...,G_k$, the following holds:
\begin{enumerate}[(i)]
	\item $G$ is rationally (or generically) identifiable if and only if all components \\ $G_1, ..., G_k$ are rationally (or generically) identifiable;
	\item $G$ is generically infinite-to-one if and only if there exists a component $G_j$ that is generically infinite-to-one;
	\item if each $G_j$ is generically $h_j$-to-one with $h_j<\infty$, then $G$ is generically $h$-to-one with $h=\prod_{j=1}^k h_j$.
\end{enumerate}
\end{thm}

We remark that this decomposition also plays a role in non-linear models; see, for instance,  the paper of \cite{shpitser2014} and the references given therein.

\section{Ancestral Decomposition} \label{sec:anc_decomp}

For a later strengthening of the HTC, we will show that the generic identification of certain subgraphs of an acyclic mixed graph $G = (V,D,B)$ implies the generic identification of their associated edge coefficients in the larger graph $G$. This result is straightforward and is well known in other forms. Surprisingly, however, this simple idea can extend the applicability of the HTC when combined with the decomposition from Theorem \ref{thm:tian}. We will first define what we mean by an ancestral subset and an induced graph,
\begin{defn}
	Let $V'\subset V$ be a subset of vertices.  The \emph{ancestors of $V'$} form the set 
\begin{align*}
	An(V') = \{v\in V: \text{there exists a directed path from $v$ to some $w\in V'$}\},
\end{align*}
where we consider the empty path to be directed so that $V'\subset An(V')$. If $V' = An(V')$, then we call $V'$ \emph{ancestral}.
\end{defn}
\begin{defn}
	Let  $V'\subset V$ be again a subset of vertices. 
	The \emph{subgraph of $G$ induced by $V'$} is the mixed graph $G_{V'}=(V',D',B')$  with
\begin{align*}
	D' &= \{v\to w\in G: v,w\in V'\}, \\
	B' &= \{v\bi w\in G: v,w\in V'\}.
\end{align*}
\end{defn}

We now have the following simple fact.

\begin{thm} \label{thm:ancestral}
Let $G=(V,D,B)$ be a mixed graph, and let $V'$ be an ancestral subset of $V$. If the induced subgraph $G_{V'}$ is generically (or rationally) identifiable then so are all the corresponding edge coefficients  in $G$.
\end{thm}

\begin{proof}
Let the covariance matrix $\Sigma = (I-\Lambda)^{-T}\Omega(I-\Lambda)^{-1}$ correspond to $G$, that is, $\Lambda\in\R^D$ and $\Omega\in PD(B)$.  Let $\Lambda'$ and $\Omega'$ denote the $V'\times V'$ submatrices of $\Lambda$ and $\Omega$, respectively, and let $\Sigma' = (I_{|V'|}-\Lambda')^{-T}\Omega(I_{|V'|}-\Lambda')^{-1}$ where $I_{|V'|}$ is the $|V'|\times |V'|$ identity matrix.
For ease of notation, write $G' = G_{V'}$.

Recall that for any $v,w\in V$, the set $\cT(v,w)$ comprises all treks between $v$ and $w$ in $G$.  Similarly, write $\cT_{G'}(v,w)$ for the set of treks between $v$ and $w$ in $G'$.  Since $V'$ is ancestral, it holds that $\cT(v,w) = \cT_{G'}(v,w)$ for all $v,w\in V'$.  Thus, by Proposition \ref{prop:trek_rule}, we have that for any $v,w\in V'$
\begin{align*}
	\Sigma_{vw} &= \sum_{\pi\in \cT(v,w)} \pi(\Lambda,\Omega) 
	= \sum_{\pi\in \cT_{G'}(v,w)} \pi(\Lambda',\Omega') 
	= \Sigma_{vw}'.
\end{align*}

Now suppose that $G'$ is generically (or rationally) identifiable. Then $\Lambda',\Omega'$ can be generically (or rationally) recovered from $\Sigma'$. As we have just shown that $\Sigma_{vw} = \Sigma_{vw}'$ for all $v,w\in V'$, we have that the entries of $\Lambda,\Omega$ corresponding to $\Lambda',\Omega'$ can be recovered from $\Sigma$ generically (or rationally).
\end{proof}

We may generalize the above theorem so that we do not have to consider the identifiability of all of $G'$ and instead only look at certain edges in $G'$.

\begin{cor}\label{cor:ancestral}
Let $G=(V,D,B)$ be a mixed graph, and let $V'$ be an ancestral subset of $V$. If an edge coefficient of $G_{V'}$ is generically (or rationally) identifiable then so is the corresponding coefficient in $G$.
\end{cor}

\begin{proof}
This follows exactly as in the proof of Theorem \ref{thm:ancestral} by only considering a single generically (or rationally) identifiable coefficient of $G'$ at a time.
\end{proof}

We give an example as to how Theorem \ref{thm:ancestral} strengthens the HTC.

\begin{figure}[t]
        \centering
        \begin{subfigure}[b]{0.5\textwidth}
\tikzset{
	every node/.style={circle, inner sep=1mm, minimum size=0.55cm, draw, thick, black, fill=white, text=black},
	every path/.style={thick}
}
\begin{tikzpicture}[align=center,node distance=2.2cm]
	\node [] (1) {1};
	\node [] (2) [right of=1]    {2};
	\node [] (3) [below right of=2]    {3};
	\node [] (4) [below left of=3]    {4};
	\node [] (5) [left of=4]    {5};
	\node [] (6) [above left of=5]    {6};
	\draw[blue] [-latex] (1) edge (2);
	\draw[blue] [-latex] (1) edge (3);
	\draw[blue] [-latex] (1) edge (6);
	\draw[blue] [-latex] (2) edge (3);
	\draw[blue] [-latex] (2) edge (4);
	\draw[blue] [-latex] (2) edge (5);
	\draw[blue] [-latex] (2) edge (6);
	\draw[blue] [-latex] (3) edge (4);
	\draw[blue] [-latex] (4) edge (5);
	\draw[red] [latex-latex] (1) edge (4);
	\draw[red] [latex-latex,bend right] (1) edge (6);
	\draw[red] [latex-latex,bend left] (2) edge (3);
	\draw[red] [latex-latex,bend left] (2) edge (5);
	\draw[red] [latex-latex,bend left] (2) edge (6);
\end{tikzpicture}
	\caption{An HTC-inconclusive graph $G$} \label{fig:htc-incon}
        \end{subfigure}%
        ~
        \begin{subfigure}[b]{0.5\textwidth}
\tikzset{
	every node/.style={circle, inner sep=1mm, minimum size=0.55cm, draw, thick, black, fill=white, text=black},
	every path/.style={thick}
}
\begin{tikzpicture}[align=center,node distance=2.2cm]
	\node [] (1) {1};
	\node [] (2) [right of=1]    {2};
	\node [] (3) [below right of=2]    {3};
	\node [] (4) [below left of=3]    {4};
	\node [] (5) [left of=4]    {5};
	\draw[blue] [-latex] (1) edge (2);
	\draw[blue] [-latex] (1) edge (3);
	\draw[blue] [-latex] (2) edge (3);
	\draw[blue] [-latex] (2) edge (4);
	\draw[blue] [-latex] (2) edge (5);
	\draw[blue] [-latex] (3) edge (4);
	\draw[blue] [-latex] (4) edge (5);
	\draw[red] [latex-latex] (1) edge (4);
	\draw[red] [latex-latex,bend left] (2) edge (3);
	\draw[red] [latex-latex,bend left] (2) edge (5);
\end{tikzpicture}
	\caption{The induced subgraph $G_{\{1,2,3,4,5\}}$} \label{fig:ancestral}
        \end{subfigure}%
        ~ \\ ~\\
         \begin{subfigure}[b]{1\textwidth}
         \centering
\tikzset{
	every node/.style={circle, inner sep=1mm, minimum size=0.55cm, draw, thick, black, fill=white, text=black},
	every path/.style={thick}
}
\begin{tikzpicture}[align=center,node distance=2.2cm]
	\node [] (1) {1};
	\node [] (2) [right of=1]    {2};
	\node [] (3) [below right of=2]    {3};
	\node [] (4) [below left of=3]    {4};
	\draw[blue] [-latex] (2) edge (4);
	\draw[blue] [-latex] (3) edge (4);
	\draw[red] [latex-latex] (1) edge (4);
\end{tikzpicture}
~ \quad
\begin{tikzpicture}[align=center,node distance=2.2cm]
	\node [] (1) {1};
	\node [] (2) [right of=1]    {2};
	\node [] (3) [below right of=2]    {3};
	\node [] (4) [below left of=3]    {4};
	\node [] (5) [left of=4]    {5};
	\draw[blue] [-latex] (1) edge (2);
	\draw[blue] [-latex] (1) edge (3);
	\draw[blue] [-latex] (2) edge (3);
	\draw[blue] [-latex] (2) edge (5);
	\draw[blue] [-latex] (4) edge (5);
	\draw[red] [latex-latex,bend left] (2) edge (3);
	\draw[red] [latex-latex,bend left] (2) edge (5);
\end{tikzpicture}
	\caption{The mixed components of $G_{\{1,2,3,4,5\}}$} \label{fig:mixed-comps}
        \end{subfigure}
        \caption{A mixed graph $G$, a subgraph induced by an ancestral subset, and the mixed components the induced graph.}\label{fig:htc-strengthening}
\end{figure}
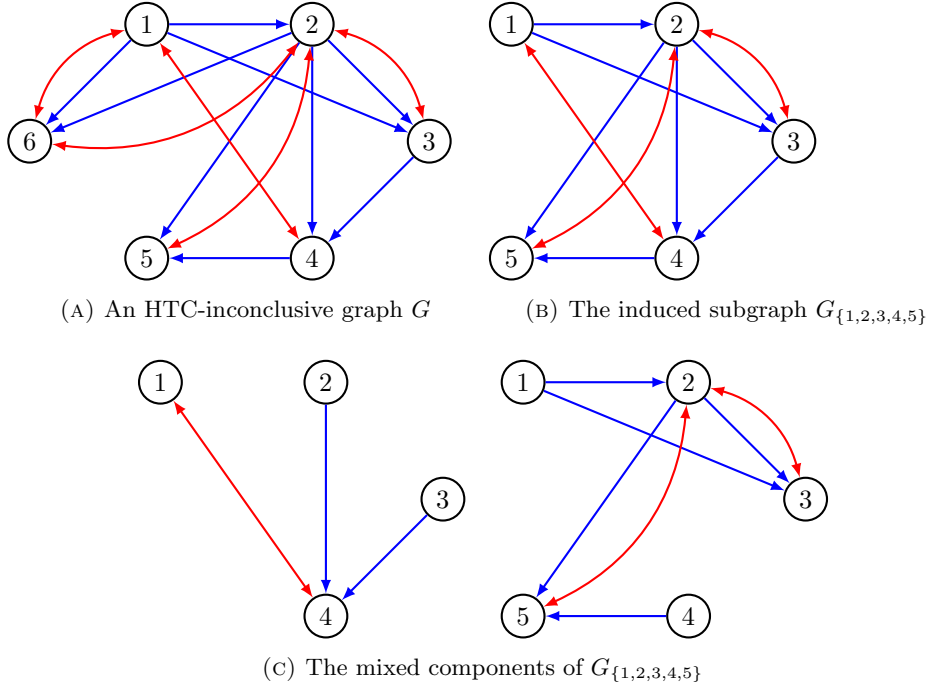

\begin{exmp}
It is straightforward to check that the graph $G$ from Figure \ref{fig:htc-incon} is HTC-inconclusive using Algorithm 1 from \cite{halftrek}. We direct the reader who does not want to perform this computation by hand to the R package SEMID \citep{semid, R}.   Moreover,  $G$ cannot be decomposed as its bidirected part is connected.

Now the set $V'=\{1,2,3,4,5\}$ is ancestral in $G$, so we may apply Theorem \ref{thm:ancestral} to the induced subgraph $G'=G_{1,2,3,4,5}$.   While $G'$ remains HTC-inconclusive, the Tian decomposition of Theorem \ref{thm:tian} is applicable. After decomposing $G'$ into its mixed components, see Figure \ref{fig:mixed-comps}, we  find that each component is HTC-identifiable and thus $G'$ itself is generically identifiable.  To show generic identifiability of $G$, we are left to show that all the coefficients on the directed edges between $pa(6)=\{1,2\}$ and 6 are generically identifiable. Since $Y = \{3,4\}$ satisfies the HTC with respect to $6$ it follows, by Lemma \ref{lem:invert} below, that $\Lambda_{pa(v),v}$ is generically identifiable. Hence,  the entire matrix $\Lambda$ is generically identifiable, and since $(I-\Lambda)^T\Sigma(I-\Lambda)=\Omega$, this implies generic identifiability of  $\Omega$.   We conclude that $G$ is generically identifiable despite being HTC-inconclusive.
\end{exmp}

\begin{lem}\label{lem:invert}
Let $G=(V,D,B)$ be a mixed graph, and let $v\in V$. If $Y\subset V$ satisfies the HTC with respect to $v$ and for each $y\in Y$ we have that $\Lambda_{pa(y),y}$ is generically (or rationally) identifiable, then $\Lambda_{pa(v),v}$ is generically (or rationally) identifiable.
\end{lem}

\begin{proof}
 Suppose vertex $v$ has $m$ parents, and let $pa(v) = \{p_1,...,p_m\}$.  Since $Y$ satisfies the HTC for $v$, we must have $|Y|=|pa(v)|=m$. Thus we may enumerate the set as $Y=\{y_1,...,y_m\}$.  Define a matrix $A\in\R^{m\times m}$ with entries
\[
	A_{ij} := \left\{
		\begin{array}{ll}
			[(I-\Lambda)^T\Sigma]_{y_ip_j} & \mbox{if } y_i\in htr(v), \\
			\Sigma_{y_ip_j} & \mbox{if } y_i\not\in htr(v),
		\end{array}
		\right.
\]
and define a vector $b\in\R^m$ with entries
\[
	b_{i} :=  \left\{
		\begin{array}{ll}
			[(I-\Lambda)^T\Sigma]_{y_iv} & \mbox{if } y_i\in htr(v), \\
			\Sigma_{y_iv} & \mbox{if } y_i\not\in htr(v).
		\end{array}
		\right.
\]
Both $A$ and $b$ are generically identifiable because we have assumed that $\Lambda_{pa(y),y}$ is generically identifiable for every $y\in Y$.  Now, from the proof of Theorem 1 in \cite{halftrek}, we have $A\cdot \Lambda_{pa(v),v}=b$, and from Lemma 2 of \cite{halftrek} we deduce that $A$ is generically invertible.  It follows that $\Lambda_{pa(v),v}=A^{-1}b$ generically so that $\Lambda_{pa(v),v}$ is generically identifiable. 
\end{proof}

Algorithm 1 from \cite{halftrek} (hereafter called the HTC-algorithm) determines whether or not a mixed graph $G=(V,D,B)$ satisfies the conditions of Theorem \ref{thm:htc-id} and and thus checks if the graph is HTCI. The HTC-algorithm operates by iteratively looping through the nodes $v\in V$ and attempting to find a half-trek system $\Pi$ to $pa(v)$ with $\Pi$ having sources which are in a set of ``allowed" nodes $A\subset V$.  Here, a node $w$ is allowed if $w\not\in htr(v)\cup\{v\} \cup sib(v)$ or if $w$ was previously shown to be generically identifiable in the sense that all coefficients on directed edges $u\to w$ were shown to be generically identifiable. If such a half-trek system $\Pi$ is found for node $v$ then \citet{halftrek} show that this implies that $v$ is generically identifiable, and thus $v$ may be added to the set of allowed nodes for the remaining iterations. The algorithm terminates when all nodes have been shown to be generically identifiable or once it has iterated through all vertices and has been unable to show the generic identifiability of any new nodes. To find a half-trek system between a suitable subset of the set of allowed nodes $A$ and $pa(v)$, the HTC-algorithm solves a Max Flow problem on an auxiliary network $G_{\text{flow}}(A,v)$, and this step takes $\cO(|V|^3)$ time when $G$ is acyclic. If in $G_{\text{flow}}(A,v)$ one can find a flow of size $|pa(v)|$ then the half-trek system exists. See Section 6 of \cite{halftrek} for more details about how $G_{\text{flow}}(A,v)$ is defined. Finally, for an acyclic mixed graph $G$, the HTC-algorithm has a worst case running time of $\cO(|V|^5)$.

Algorithm \ref{alg:id} presents a simple modification of the HTC-algorithm to leverage Corollary \ref{cor:ancestral} extending the ability of the HTC to determine the generic identifiability of acyclic mixed graphs.  We emphasize that this algorithm considers only certain ancestral subsets and, as such, we do not necessarily expect the algorithm to reach a conclusion in all cases in which Corollary \ref{cor:ancestral} may be applicable.

\begin{algorithm}[t]
\caption{A sufficient test for generic identifiability.}
\label{alg:id}
\begin{algorithmic}[1]
\STATE\textbf{Input:} $G=(V,D,B)$, an acyclic mixed graph on $n$ nodes $v_1,\dots,v_n$\\
\STATE\textbf{Initialize:} $\text{SolvedNodes}\leftarrow\{v: pa
(v)=\varnothing\}$.\\
\REPEAT
	\FOR{$v=v_1,v_2,\ldots,v_n$}
		\IF{$v\not\in$ SolvedNodes}
			\STATE \(\triangleright\) Check if we can generically identify $\Lambda_{pa(v),v}$ using
			\STATE \(\triangleright\) the induced graph $G_{An(\{v\}\cup (A\cap S)}$.
			\STATE $S \gets An(v)\cup sib(An(v))$
			\STATE $A\leftarrow S\cap ($SolvedNodes $\cup\ (V\setminus htr(v)) )\setminus(\{v\}\cup sib(v))$
			\STATE $G' \gets $ the mixed component of $G_{An(\{v\}\cup (A\cap S))}$ containing $v$
			\STATE $A\gets (A$ $\cap$ (vertices in $G')$) $\cup$ (source nodes in $G'$)
			\IF{MaxFlow$(G_{\mathrm{flow}}(v,A))=|pa(v)|$}
				\STATE SolvedNode $\gets$ SolvedNodes $\cup\ \{v\}$.
				\STATE Skip to next iteration of the loop
			\ENDIF
			\STATE \(\triangleright\) Check if we can generically identify $\Lambda_{pa(v),v}$ using
			\STATE \(\triangleright\) the induced graph $G_{An(\{v\})}$.
			\STATE $G' \gets$ the mixed component of $G_{An(v)}$ containing $v$ \label{line:start}
			\STATE $A \gets (A\cap ($vertices of $G'$)) $\cup\ ($source nodes in $G'$)
			\IF{MaxFlow$(G'_{\mathrm{flow}}(v,A))=|pa(v)|$}
				\STATE SolvedNode $\gets$ SolvedNodes $\cup\ \{v\}$.
			\ENDIF  \label{line:end}
		\ENDIF
	\ENDFOR
\UNTIL{SolvedNodes $=V$ or no change has occurred in the last iteration.}
\STATE\textbf{Output:} ``yes" if SolvedNodes $=V$, ``no" otherwise.
\end{algorithmic}
\end{algorithm}

\begin{prop}
Algorithm \ref{alg:id} returns ``yes" only if the input acyclic mixed graph $G$ is generically identifiable and will return ``yes" whenever the HTC-algorithm does. Moreover, Algorithm \ref{alg:id} returns ``yes" for the HTC-inconclusive graph in Figure \ref{fig:htc-incon} and has time complexity at most $\cO(|V|^5)$.
\end{prop}

\begin{proof}
The fact that Algorithm \ref{alg:id} only returns ``yes" if $G$ is generically identifiable follows from Theorem 7 in \cite{halftrek,halftrek-supp} and our Corollary \ref{cor:ancestral}. That Algorithm \ref{alg:id}  returns ``yes" whenever the HTC-algorithm does can be argued as follows:

 If, for a set of allowed nodes $A$ and $v\in V$, there exists $Y\subset A$ satisfying the HTC for $v$ then we must have that $Y\subset S:=An(v)\cup sib(An(v))$ and $Y$ satisfies the HTC for $v$ in $G_{An(\{v\}\cup (A\cap S))}$.  Lemma 4 of \cite{halftrek-supp} then yields that there exists a set $Y'$ of allowed nodes which satisfies the HTC for $v$ in the mixed component of $G_{An(\{v\}\cup (A\cap S))}$ containing $v$. Hence, if $v$ is added to the set of solved nodes in the HTC-algorithm it will also be added to the set of solved nodes in Algorithm \ref{alg:id}. From this it follows that if the HTC-algorithm outputs ``yes'' then so will Algorithm \ref{alg:id}.
	
It is straightforward to check that Algorithm \ref{alg:id} returns ``yes'' for the graph in Figure \ref{fig:htc-incon} and, thus, it remains only to argue that the time complexity is at most $\cO(|V|^5)$. Note that the Max Flow algorithm for this problem has a running time of $\cO(|V|^3)$ since $G$ is acyclic, see \citet{halftrek} for details. It is easy to see that this running time dominates on each iteration. Moreover, since at the end of each iteration through the $|V|$ nodes of the graph,  the algorithm must either terminate or add at least one vertex to the set of solved nodes, it follows that there are at most $|V|^2$ iterations.   We conclude that  the maximum run time of the algorithm is  $\cO(|V|^2\cdot |V|^3) = \cO(|V|^5)$.
\end{proof}

\begin{rem}
One might expect that lines \ref{line:start} to \ref{line:end} in Algorithm \ref{alg:id} are superfluous. This is, however, false, and indeed we have found examples of graphs $G$ on 10 nodes for which Algorithm \ref{alg:id} returns ``yes'' but the corresponding algorithm with lines \ref{line:start} to \ref{line:end} removed returns ``no.'' As these examples are fairly large we have chosen to not display them here.
\end{rem}

\section{Computational Experiments} \label{sec:comp}

We now run a simulation study to examine the effect of applying Algorithm \ref{alg:id} to HTC-inconclusive graphs. All code is written in $R$, and we use the SEMID package to determine HTC-identifiability and HTC-unidentifiability \citep{R,semid}.
\begin{algorithm}[t]
\caption{A procedure for generating random acyclic mixed graphs.} \label{alg:rgraph}
\begin{algorithmic}[1]
\STATE\textbf{Input:} A positive integer $n$ and $0<p,q<1$\\
\STATE\textbf{Initialize:} A mixed graph $G=(V,D,B)$ with $V=\{1,...,n\}$, $D=B=\emptyset$
\STATE Pick  a random collection $E$ of $n-1$ bidirected edges so that $(V,E)$ is a tree.
\STATE $B \gets E$
\FOR{$1\leq i< j \leq n$}
	\STATE Add $i\bi j$ to $B$ with probability $p$
\ENDFOR
\FOR{$1\leq i< j \leq n$}
	\STATE Add $i\to j$ to $D$ with probability $q$
\ENDFOR
\STATE\textbf{Output:} G
\end{algorithmic}
\end{algorithm}
For each combination of $n\in\{6,8,10,12\}$, $p\in\{.1,.2,.3\}$, and $q\in\{.2,.3,.4,.5,.6\}$ we perform the following steps:
\begin{enumerate}[(i)]
	\item Use Algorithm \ref{alg:rgraph} with probability parameters $p$ and $q$ to generate random acyclic mixed graphs $G$ with connected bidirected part on $n$ nodes, until we have found 1000 graphs which are HTC-inconclusive.
	\item For each of the 1000 HTC-inconclusive graphs $G$, use Algorithm \ref{alg:id} to test the generic identifiability of $G$.
	\item Record the proportion of the 1000 graphs that are shown to be generically identifiable by Algorithm \ref{alg:id}.  Call this proportion $a_{n,p,q}$.
\end{enumerate}
To summarize our findings we compute, for each pair $(n,p)$, the average $b_{n,q} = \frac{1}{3}\sum_{p} a_{n,p,q}$.   We then plot the values of $b_{n,q}$  in Figure \ref{fig:results_graph}.   According to this figure, Algorithm \ref{alg:id} provides a modest increase in the number of graphs that are generically identifiable. This improvement is seen to be largest when $q$ is large, that is, the directed part of the mixed graph is dense.

\begin{figure}[t]
        \centering
        \centerline{\includegraphics[scale=.45]{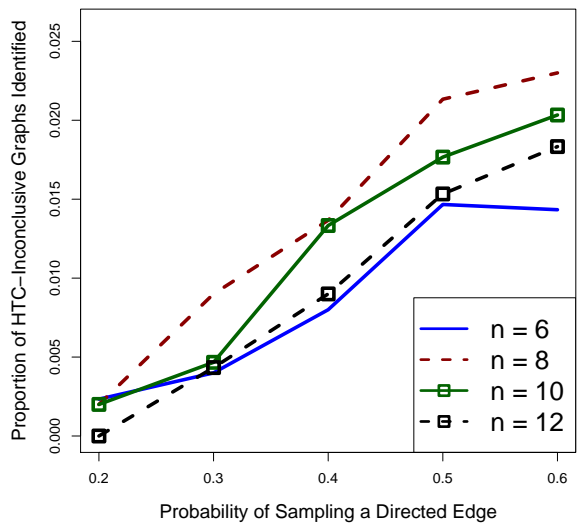}}
        \caption{The average proportion of HTC-inconclusive graphs found to be generically identifiable by Algorithm \ref{alg:id}.}\label{fig:results_graph}
\end{figure}

\section{Conclusion} \label{sec:conclusion}

We have shown how the generic identifiability of a subgraph of a mixed graph $G$ induced by an ancestral subset of vertices implies the generic identifiability of the corresponding  edge coefficients in $G$ (Theorem \ref{thm:ancestral} and Corollary \ref{cor:ancestral}). We then provided, in Algorithm \ref{alg:id}, one specific way of how to leverage this result by using the HTC of \cite{halftrek} and the decomposition techniques of \citet{tian}.  Our new algorithm provides a modest strengthening of the HTC while not increasing the algorithmic complexity of the HTC-algorithm of \citet{halftrek}.   

When saying above that Algorithm \ref{alg:id} constitutes one way of leveraging, we mean that the algorithm considers only certain ancestral subsets.  While we do not have any examples to report, it is possible that there are acyclic mixed graphs for which Algorithm \ref{alg:id} does not return ``yes'' but which could be proven generically identifiable by a combination of Corollary \ref{cor:ancestral}, the Tian decomposition, and the HTC-algorithm.  This said, it is not clear to us that all ancestral subsets can be considered in an algorithm with polynomial run time.  Clarifying this issue would be an interesting topic for future work.

%\appendix

\section*{Acknowledgments}

We would like to thank Thomas Richardson whose questions following a seminar talk started this project.  This work was partially supported by the U.S. National Science Foundation (DMS-1305154) and National Security Agency (H98230-14-1-0119). 
%The United States Government is authorized to reproduce and distribute reprints.

\bibliographystyle{abbrvnat}
\bibliography{half_trek_ancestral}
\end{document}